\documentclass[11pt]{article}
\usepackage{amsmath,amssymb}
\usepackage{graphicx}
\usepackage{amsfonts}

\usepackage{amsthm}
\newtheorem{thm}{Theorem}[section]

\usepackage[USenglish]{babel} %francais, polish, spanish, ...
\usepackage[T1]{fontenc}
\usepackage[ansinew]{inputenc}
\usepackage{mathrsfs}
\usepackage{stmaryrd} % for \boxbar
\usepackage{epstopdf} % for adding eps figures
\usepackage{verbatim}
\usepackage{MnSymbol}  %for << and >>, that is \llangle \rrangle
\usepackage{accents}   %for making \dddot look nice
\usepackage{bbm}			 % for \mathbbm{}; makes 1 the mathbb style
\usepackage{authblk}   % for multiple affiliations of the authors

\numberwithin{equation}{section}
\numberwithin{figure}{section}

\title{New variational and multisymplectic formulations of the Euler-Poincar\'e equation on the Virasoro-Bott group using the inverse map}
\date{}
\author[1]{Darryl D. Holm\thanks{\texttt{d.holm@imperial.ac.uk}}}
\author[1,2]{Tomasz M. Tyranowski\thanks{\texttt{tomasz.tyranowski@ipp.mpg.de}}}

\affil[1]{\small Mathematics Department\authorcr Imperial College London \authorcr London SW7 2AZ, UK}
\affil[2]{\small Max-Planck-Institut f\"ur Plasmaphysik \authorcr Boltzmannstra{\ss}e 2, 85748 Garching, Germany }

\begin{document}

\maketitle

\begin{abstract}
We derive a new variational principle, leading to a new momentum map and a new multisymplectic formulation for a family of Euler--Poincar\'e equations defined on the Virasoro-Bott group, by using the inverse map (also called `back-to-labels' map). This family contains as special cases the well-known Korteweg-de Vries, Camassa-Holm, and Hunter-Saxton soliton equations. In the conclusion section, we sketch opportunities for future work that would apply the new Clebsch momentum map with $2$-cocycles derived here to investigate a new type of interplay among nonlinearity, dispersion and noise. 
\end{abstract}

%%%%%%%%%%%%%%%%%%%%%%%%%%%%%%%%%%%%%%%%%%%%%%%%%%%%%%%%%%%%%%%%%%%%%%%%%%%%%%%%%%%%
%  INTRODUCTION
%%%%%%%%%%%%%%%%%%%%%%%%%%%%%%%%%%%%%%%%%%%%%%%%%%%%%%%%%%%%%%%%%%%%%%%%%%%%%%%%%%%%
\section{Introduction}
\label{sec:intro}

The family of equations

\begin{equation}
\label{eq: General family of equations}
\alpha (u_t+3 u u_x)-\beta (u_{xxt}+2 u_x u_{xx}+u u_{xxx}) + a u_{xxx}=0,
\end{equation}

\noindent
where $a$, $\alpha$, $\beta$ are real nonnegative parameters, was introduced in \cite{KhesinMisiolek2003} as the geodesic flow dynamics associated to a variety of right-invariant metrics on the Virasoro-Bott group (see also \cite{Khesin2008Geometry}, \cite{MisiolekKdV}). Various well-known completely integrable soliton equations are special cases of \eqref{eq: General family of equations}. For example, when $\alpha=1$ and $\beta=0$, then equation \eqref{eq: General family of equations} specialises to the Korteweg-de Vries equation (\cite{KortewegDeVries}, \cite{DrazinSolitonsIntroduction})

\begin{equation}
\label{eq: KdV equation}
u_t + 3 u u_x + a u_{xxx}=0\,;
\end{equation}

\noindent
whereas for $\alpha=\beta=1$ one obtains the Camassa-Holm equation (\cite{CamassaHolm}, \cite{CamassaHolmHyman}, \cite{HolmIvanovReview})

\begin{equation}
\label{eq: C-H equation}
u_t-u_{xxt}+3 u u_x-2 u_x u_{xx}-u u_{xxx} + a u_{xxx}=0\,;
\end{equation}

\noindent
and for $\alpha=0$ and $\beta=1$ one finds the Hunter-Saxton equation (\cite{HunterSaxton1991}, \cite{HunterZheng1994}),

\begin{equation}
\label{eq: H-S equation}
u_{xxt}+2 u_x u_{xx}+u u_{xxx}- a u_{xxx}=0.
\end{equation}

The aim of this paper is to derive a new canonical variational principle for the family of equations \eqref{eq: General family of equations}, and thereby determine its new multisymplectic formulation. By doing so, we obtain unified variational and multisymplectic characterizations of all three of the well-known integrable soliton equations, KdV, CH, and HS, which are subcases of the general family of equations in \eqref{eq: General family of equations}.

Variational principles have proved extremely useful in the study of nonlinear evolution PDEs. For instance, they often provide physical insights into the problem being  considered; facilitate discovery of conserved quantities by relating them to symmetries via Noether's theorem; allow one to determine approximate solutions to PDEs by minimizing the action functional over a class of test functions (see, e.g., \cite{Cooper1993}); and provide a way to construct a class of numerical methods called variational integrators (see \cite{MarsdenPatrickShkoller}, \cite{MarsdenWestVarInt}). A canonical variational principle for the KdV equation expressed in terms of the velocity potential was first proposed by Whitham \cite{WhithamVariational}; see also \cite{Cooper1993}, \cite{DrazinSolitonsIntroduction}, \cite{Khater2003}, \cite{MarsdenRatiuSymmetry}. In fact, there is an infinite family of such Lagrangians, as shown by Nutku \cite{Nutku2001}. Two canonical variational principles for the dispersionless CH equation ($a=0$) were introduced in \cite{CotterHolmMultisymplectic} and \cite{KouranbaevaShkoller}. Two variational structures are also known for the HS equation with $a=0$ (see \cite{AliHunter2009}, \cite{HunterSaxton1991}, \cite{HunterZheng1994}).

Multisymplectic structures of Hamiltonian PDEs were first considered by Bridges \cite{BridgesMultisymplectic} as a natural generalization of the symplectic structure of Hamiltonian ODEs. Among other applications, the multisymplectic formalism is useful for, e.g., the stability analysis of water waves (see \cite{BridgesMultisymplectic}, \cite{BridgesDerks1999}) and construction of a class of numerical methods known as multisymplectic integrators (see \cite{BridgesReichMultisymplectic}, \cite{MarsdenPatrickShkoller}). It has been  observed in the literature that, as for symplectic integrators for Hamiltonian ODEs, multisymplectic integrators demonstrate superior performance in capturing long time dynamics of PDEs (see \cite{MooreReichBEA}). To the best of our knowledge, only one multisymplectic formulation of the KdV equation has been considered so far (see \cite{BridgesDerks1999}, \cite{ZhaoQin2000}). Four different multisymplectic formulations are known for the dispersionless CH equation (see \cite{CohenMultisymplecticCH}, \cite{CotterHolmMultisymplectic}, \cite{KouranbaevaShkoller}). Two multisymplectic structures for the HS equation with $a=0$ were described in \cite{MiyatakeCohen2017}. 

\paragraph{Main content}
The main content of the remainder of this paper is, as follows. 
\begin{description}
\item
In Section~\ref{sec:The inverse map and Clebsch representation} we review the Euler-Poincar\'e theory on the Virasoro-Bott group and then construct a new canonical variational principle in terms of the inverse map. The main result of this section is Theorem~\ref{thm: Elimination theorem}, which provides the Clebsch variational principle for Euler-Poincar\'e equations on the dual of the Virasoro-Bott algebra. The variational equations yield the new Clebsch momentum map $T^*\widehat{\text{Diff}}(S^1)\to \mathfrak{vir}^*$ in \eqref{momap} associated with particle relabeling by cotangent-lifted right actions of $\widehat{\text{Diff}}(S^1)$ that include the Bott $2$-cocycle. Section~\ref{sec:The inverse map and Clebsch representation} also introduces the simplified Clebsch variational principle \eqref{eq: Simplified Clebsch action functional}, which yields the special family of equations in \eqref{eq: General family of equations} as its Euler-Lagrange equation.
\item
In Section~\ref{sec:Inverse map multisymplectic formulation} we use the Clebsch representation based on the inverse flow map to derive the multisymplectic form formula associated with our variational principle. We then deduce a new multisymplectic formulation of the family of equations \eqref{eq: General family of equations}. The main result of this section is Theorem~\ref{thm: Multisymplectic PDE system}, which derives the multisymplectic formulation based on the  inverse flow map. 
\item
Section~\ref{sec:Summary} contains the summary of the present work and a discussion of several new directions for research that it reveals.
\end{description}

%%%%%%%%%%%%%%%%%%%%%%%%%%%%%%%%%%%%%%%%%%%%%%%%%%%%%%%%%%%%%%%%%%%%%%%%%%%%%%%%%%%%
%  The inverse map and Clebsch representation
%%%%%%%%%%%%%%%%%%%%%%%%%%%%%%%%%%%%%%%%%%%%%%%%%%%%%%%%%%%%%%%%%%%%%%%%%%%%%%%%%%%%
\section{The inverse map and Clebsch representation}
\label{sec:The inverse map and Clebsch representation}

Equation~\eqref{eq: General family of equations} was first introduced in the Lie-Poisson context (see \cite{KhesinMisiolek2003}, \cite{Khesin2008Geometry}, \cite{MisiolekKdV}). In this section we take the Lagrangian point of view, instead, and formulate \eqref{eq: General family of equations} as an Euler-Poincar\'e equation on the Virasoro-Bott group. For this, we construct a canonical variational principle that will later allow us to determine a multisymplectic formulation of \eqref{eq: General family of equations}.

\subsection{Euler-Poincar\'e equation on the Virasoro-Bott group}
\label{sec: Euler-Poincare equation on the Virasoro-Bott group}

Let $S^1=\mathbb{R}/2\pi \mathbb{Z}=\{\theta \in [0,2\pi)\}$ denote the circle group, and let $\text{Diff}(S^1)$ be the diffeomorphism group of $S^1$. The tangent bundles can be identified as $TS^1 = S^1 \times \mathbb{R}$ and $T\text{Diff}(S^1) = \text{Diff}(S^1) \times \mathfrak{X}(S^1)$, where $\mathfrak{X}(S^1)=\{ \chi: S^1 \longrightarrow \mathbb{R}\}$ is the set of all smooth vector fields on $S^1$. In particular, the Lie algebra of $S^1$ is $\mathbb{R}$, and the Lie algebra of $\text{Diff}(S^1)$ is $\mathfrak{X}(S^1)$. The Virasoro-Bott group is the central extension $\widehat{\text{Diff}}(S^1)=\text{Diff}(S^1) \times S^1$ with the group operation

\begin{equation}
\label{eq: Virasoro-Bott group operation}
(\psi_1, \theta_1)\cdot(\psi_2, \theta_2) = (\psi_1\circ \psi_2,B(\psi_1,\psi_2)+\theta_1+\theta_2),
\end{equation}

\noindent
where the 2-cocycle $B(\psi_1,\psi_2)$ is given by

\begin{equation}
\label{eq: Bott cocycle}
B(\psi_1,\psi_2) = \frac{1}{2}\int_{S^1} \log \frac{\partial (\psi_1\circ \psi_2)}{\partial x}\,d\log \frac{\partial \psi_2}{\partial x}.
\end{equation}

\noindent
The tangent bundle of the Virasoro-Bott group is $T\widehat{\text{Diff}}(S^1) = \widehat{\text{Diff}}(S^1) \times \mathfrak{X}(S^1) \times \mathbb{R}$. The Virasoro algebra $\mathfrak{vir}$ is the Lie algebra of the Virasoro-Bott group and it can be identified as $\mathfrak{vir} = \mathfrak{X}(S^1) \times \mathbb{R}$. The Lie algebra bracket (or adjoint action) on $\mathfrak{vir}$ is given by

\begin{equation}
\label{eq: Lie algebra bracket}
\text{ad}_{(u,a)}(v,b) = \big[(u,a),(v,b)\big] = \bigg( -u v_x + u_x v, \int_{S^1} u_x v_{xx}\,dx \bigg)
\end{equation}

\noindent
for $(u,a), (v,b) \in \mathfrak{vir}$. One identifies the dual of $\mathfrak{vir}$ with itself, for the $L^2$ inner product

\begin{equation}
\label{eq: L2 inner product}
\big\langle (u,a), (v,b) \big\rangle = ab+\int_{S^1} u v\,dx.
\end{equation}

\noindent
With respect to this inner product, the coadjoint action $\text{ad}^*_{(u,a)}: \mathfrak{vir} \longrightarrow \mathfrak{vir}$ can be represented as

\begin{equation}
\label{eq: Coadjoint action}
\text{ad}^*_{(u,a)} (v,b) = (2 v u_x + u v_x + b u_{xxx},0).
\end{equation}

\noindent
For more information about the Virasoro-Bott group and the Virasoro algebra we refer the reader to \cite{Khesin2008Geometry} and \cite{MarsdenRatiuSymmetry}.

Suppose a Lagrangian system is defined on $T\widehat{\text{Diff}}(S^1)$ by specifying the right-invariant Lagrangian $L: T\widehat{\text{Diff}}(S^1) \longrightarrow \mathbb{R}$. Rather then on the full tangent bundle, the dynamics of such a system can be analyzed on the Lie algebra $\mathfrak{vir}$ via the process called Euler-Poincar\'e reduction (see \cite{HoMaRaAIM1998}, \cite{HolmGMS}, \cite{MarsdenRatiuSymmetry}). We consider the reduced Lagrangian $\ell: \mathfrak{vir}\longrightarrow \mathbb{R}$ defined by $\ell(u,a) = L(\text{id},0,u,a)$ and the reduced variational principle

\begin{equation}
\label{eq: E-P variational principle}
\delta \int_{t_1}^{t_2} \ell\big( u(t),a(t) \big)\,dt =0,
\end{equation} 

\noindent
using variations of the form $\delta (u,a) = \frac{d}{dt} (v,b) - \big[(u,a),(v,b)\big]$, where $(v(t),b(t))$ vanish at the endpoints, and the time derivative of a Virasoro algebra-valued function of time is understood as $\frac{d}{dt} (v(t),b(t)) = \big(\frac{\partial v}{\partial t}(\cdot,t), \frac{db}{dt}(t)\big)$. This variational principle leads to the \emph{Euler-Poincar\'e equation},

\begin{equation}
\label{eq: E-P equation}
\frac{d}{dt} \frac{\delta \ell}{\delta (u,a)} + \text{ad}^*_{(u,a)} \frac{\delta \ell}{\delta (u,a)} = 0,
\end{equation}

\noindent
where the variational derivatives and the coadjoint action are computed with respect to the inner product \eqref{eq: L2 inner product}. Below, we demonstrate that \eqref{eq: General family of equations} can be written as an Euler-Poincar\'e equation.

\begin{thm}
\label{thm: The EP form of the family of equations}
Let the reduced Lagrangian be defined as

\begin{equation}
\label{eq: Reduced Lagrangian for the family of equations}
\ell(u,a) = \frac{1}{2}a^2 + \frac{1}{2} \int_{S^1}\big(\alpha u^2 + \beta u_x^2 \big)\,dx,
\end{equation}

\noindent
where $\alpha,\beta\geq 0$. Then the corresponding Euler-Poincar\'e equations take the form

\begin{align}
\label{eq: E-P form of the family of equations}
\alpha (u_t+3 u u_x)-\beta (u_{xxt}+2 u_x u_{xx}+u u_{xxx}) + a u_{xxx}=0,
\quad\,\,\hbox{and}\,\,\quad
\frac{da}{dt}=0.
\end{align}
\end{thm}

\begin{proof}
The case $\alpha=1$ and $\beta=0$ is shown in \cite{MarsdenRatiuSymmetry}. The case $\alpha,\beta\geq 0$ is a straightforward generalization.
\end{proof}

\noindent
The first equation in \eqref{eq: E-P form of the family of equations} is equivalent to \eqref{eq: General family of equations}; since the second equation in \eqref{eq: E-P form of the family of equations} implies $a=\text{const}$.

\subsection{Reconstruction equations and the inverse map}
\label{sec: Reconstruction equations and the inverse map}

A solution $(u(t),a(t))$ of \eqref{eq: E-P equation} describes the evolution of the (right-invariant) Lagrangian system in the Virasoro algebra, denoted $\mathfrak{vir}$. One can reconstruct the evolution on the whole Virasoro-Bott group by finding a curve $(\psi(t),\theta(t)) \in \widehat{\text{Diff}}(S^1)$ which right-translates its tangent vector back to $(u(t),a(t))$, i.e., in short-hand notation $(u(t),a(t)) =(\dot \psi(t),\dot \theta(t))\cdot(\psi(t),\theta(t))^{-1}$. More precisely,

\begin{equation}
\label{eq: Right translation back to the Lie algebra}
\big(u(t),a(t)\big) \cong \big(\text{id},0,u(t),a(t)\big) = T_{(\psi(t),\theta(t))}R_{(\psi^{-1}(t),-\theta(t))}\cdot \big(\psi(t),\theta(t),\dot \psi(t),\dot \theta(t)\big),
\end{equation}

\noindent
where $R$ denotes right translation on the Virasoro-Bott group and $TR$ its tangent lift (see \cite{HolmGMS}, \cite{MarsdenRatiuSymmetry}). By using \eqref{eq: Virasoro-Bott group operation} and \eqref{eq: Bott cocycle}, we obtain the reconstruction equations

\begin{align}
\label{eq: Reconstruction equations}
\begin{split}
u(t) &= \dot \psi(t)\circ \psi^{-1}(t),  \\
a(t) &= \dot \theta(t) + \frac{d}{ds}\bigg|_{s=t} B\big(\psi(s),\psi^{-1}(t)\big).
\end{split}
\end{align}

In the context of fluid dynamics, a time-dependent diffeomorphism $\psi(t) \in \text{Diff}(S^1)$ maps a given reference configuration to the fluid domain at each instant of time, i.e., $\psi(X,t)$ represents the position at time $t$ of the fluid particle labeled by $X$. On the other hand, the inverse map $l(t)=\psi^{-1}(t)$ maps from the current configuration of the fluid to the reference configuration, i.e., $l(x,t)$ is the label of the fluid particle occupying the position $x$ at time $t$. The Eulerian velocity field $u(x,t)$ gives the velocity of the fluid particle that occupies the position $x$ at time $t$, i.e., $\dot \psi(X,t) = u(\psi(X,t),t)$. This is precisely the meaning of the first of the reconstruction equations in \eqref{eq: Reconstruction equations}. It will be convenient for us to rewrite the reconstruction equations in terms of the inverse map. One can check via a straightforward calculation that the first equation in \eqref{eq: Reconstruction equations} is equivalent to

\begin{equation}
\label{eq: Advection equation for the inverse map}
l_t + u l_x = 0.
\end{equation}

\noindent
Using the definition of the 2-cocycle \eqref{eq: Bott cocycle}, we further calculate

\begin{align}
\label{eq: Derivative of the cocycle term}
\frac{d}{ds}\bigg|_{s=t} B\big(\psi(s),\psi^{-1}(t)\big)&=\frac{1}{2}\int_{S^1} \frac{\partial \big(\dot \psi(t)\circ \psi^{-1}(t)\big)}{\partial x}\,d\log\frac{\partial \psi^{-1}(t)}{\partial x} \nonumber\\
&=\frac{1}{2}\int_{S^1} u_x\,d\log l_x \\
&= \frac{1}{2}\int_{S^1} \frac{u_x l_{xx}}{l_x}\,dx, \nonumber
\end{align}

\noindent
where in deriving the second equality we have used the first reconstruction equation in \eqref{eq: Reconstruction equations} and the definition of the inverse map. Therefore, the reconstruction equations in terms of the inverse map take the forms

\begin{align}
\label{eq: Reconstruction equations in terms of the inverse map}
l_t + u l_x = 0,
\,\,\quad\hbox{and}\quad\,\,
a(t) = \dot \theta(t) + \frac{1}{2}\int_{S^1} \frac{u_x l_{xx}}{l_x}\,dx.
\end{align}

\noindent
Thus, given a solution $(u(t),a(t))$ of \eqref{eq: E-P equation}, one can easily solve \eqref{eq: Reconstruction equations in terms of the inverse map} for $l(x,t)$ and $\theta(t)$.

\subsection{Clebsch variational principle}
\label{sec: Clebsch variational principle}

\subsubsection{General reduced Lagrangian}
\label{sec: General reduced Lagrangian}

As discussed in Section~\ref{sec: Euler-Poincare equation on the Virasoro-Bott group}, Equation~\eqref{eq: General family of equations} has an underlying variational structure. However, the Euler-Poincar\'e variational principle \eqref{eq: E-P variational principle} imposes constraints on the variations of the functions $u$ and $a$, which may be inconvenient in some applications, for instance, when one is interested in deriving variational integrators, or determining the underlying multisymplectic structure, as is our goal in this work. One can circumvent this issue by considering an augmented action functional which includes the reconstruction equations as constraints. This idea was formalized in the context of variational Lie group integrators in back-to-back papers in \cite{BouRabeeHP} and \cite{CotterHolm-ClebschInt}. The idea of using the inverse map $l(x,t)$ (also called `back-to-labels' map) and the advection condition \eqref{eq: Advection equation for the inverse map} appeared in \cite{Holm1983,HoKuLe1983}, and was later used in \cite{CotterHolmMultisymplectic} to construct multisymplectic formulations of a class of fluid dynamics equations. We extend these ideas to systems defined on the Virasoro-Bott group.

The Clebsch variational principle (also sometimes called the Hamilton-Pontryagin principle) enforces stationarity of the action $S=\int \ell(u,a)\,dt$ under the constraint that the reconstruction equations \eqref{eq: Reconstruction equations in terms of the inverse map} are satisfied. We define the augmented action functional

\begin{equation}
\label{eq: Clebsch action functional}
S[u,a,l,\theta,\pi,\lambda] = \int_{t_1}^{t_2}\ell(u,a)\,dt + \int_{t_1}^{t_2} \int_{S^1} \pi(l_t+ul_x)\,dxdt + \int_{t_1}^{t_2} \lambda \bigg(\dot \theta - a + \frac{1}{2} \int_{S^1} \frac{u_x l_{xx}}{l_x}\,dx \bigg)\,dt,
\end{equation}

\noindent
where $\pi=\pi(x,t)$ and $\lambda=\lambda(t)$ are Lagrange multipliers, and consider the variational principle

\begin{equation}
\label{eq: Clebsch variational principle}
\delta S = 0,
\end{equation}

\noindent
with respect to arbitrary variations $\delta u$, $\delta a$, $\delta \pi$, $\delta \lambda$, and vanishing endpoint variations $\delta l$ and $\delta \theta$, i.e., $\delta l(x,t_1) = \delta l(x, t_2) = \delta \theta(t_1) = \delta \theta(t_2) =0$. The resulting variational equations are

\begin{subequations}
\label{eq: E-L equations for the Clebsch action functional}
\begin{align}
\label{eq: E-L equations for the Clebsch action functional 1}
\delta \theta&: \quad \dot \lambda = 0, \\
\label{eq: E-L equations for the Clebsch action functional 2}
\delta a&:      \quad  \frac{\partial \ell(u,a)}{\partial a} = \lambda, \\
\label{eq: E-L equations for the Clebsch action functional 3}
\delta \lambda&:\quad \dot \theta = a - \frac{1}{2} \int_{S^1} \frac{u_x l_{xx}}{l_x}\,dx,\\
\label{eq: E-L equations for the Clebsch action functional 4}
\delta u&:      \quad \frac{\delta \ell(u,a)}{\delta u}  = -\pi l_x +  \frac{\lambda}{2} \frac{\partial}{\partial x} \frac{l_{xx}}{l_x},\\
\label{eq: E-L equations for the Clebsch action functional 5}
\delta \pi&:    \quad l_t + u l_x =0, \\
\label{eq: E-L equations for the Clebsch action functional 6}
\delta l&:      \quad \pi_t + \frac{\partial}{\partial x} \bigg(\pi u -\frac{\lambda}{2}  \frac{u_{xx}}{l_x}  \bigg)=0\,.
\end{align}
\end{subequations}

\noindent
Consequently, we obtain the components of the Clebsch momentum map, given by
\begin{equation}
\label{momap}
\frac{\delta \ell}{\delta (u,a)} = \left(\frac{\delta \ell}{\delta u},\frac{\partial \ell}{\partial a}\right)
= \left( -\,\pi l_x + \frac{\lambda}{2}  \frac{\partial}{\partial x}  \frac{l_{xx}}{l_x}\,,\,\lambda\right)
\,.
\end{equation}

\paragraph{Remark.} Equation \eqref{momap} is the Clebsch momentum map $T^*\widehat{\text{Diff}}(S^1)\to \mathfrak{vir}^*$ associated with particle relabeling by cotangent-lifted right actions of $\widehat{\text{Diff}}(S^1)$ that include the Bott $2$-cocycle in equation \eqref{eq: Bott cocycle}.\\

We will now show that the dynamics generated by the system \eqref{eq: E-L equations for the Clebsch action functional} is equivalent to the dynamics generated by the Euler-Poincar\'e equation \eqref{eq: E-P equation}.

\begin{thm}
\label{thm: Elimination theorem}
Suppose the functions $u(x,t)$, $a(t)$, $l(x,t)$, $\theta(t)$, $\pi(x,t)$, and $\lambda(t)$ satisfy the Euler-Lagrange equations \eqref{eq: E-L equations for the Clebsch action functional}. Then the functions $u(x,t)$ and $a(t)$ satisfy the Euler-Poincar\'e equation \eqref{eq: E-P equation}.
\end{thm}

\begin{proof}
Let $(w,c)$ be an arbitrary element of the Virasoro algebra $\mathfrak{vir}$. Let us calculate

\begin{equation}
\bigg \langle\frac{d}{dt} \frac{\delta \ell}{\delta (u,a)}, (w,c) \bigg \rangle = \int_{S^1} \bigg(\frac{\partial}{\partial t} \frac{\delta \ell}{\delta u}\bigg)\cdot w\,dx + \bigg(\frac{\partial}{\partial t} \frac{\partial \ell}{\partial a}\bigg)\cdot c,
\end{equation}

\noindent
where the inner product $\langle \cdot,\cdot \rangle$ was defined in \eqref{eq: L2 inner product}. By using \eqref{eq: E-L equations for the Clebsch action functional 1}, \eqref{eq: E-L equations for the Clebsch action functional 2}, and \eqref{eq: E-L equations for the Clebsch action functional 3}, we further have

\begin{align}
\bigg \langle\frac{d}{dt} \frac{\delta \ell}{\delta (u,a)}, (w,c) \bigg \rangle &= \int_{S^1} \frac{\partial}{\partial t} \bigg(\frac{1}{2} \lambda \frac{\partial}{\partial x} \frac{l_{xx}}{l_x}-\pi l_x \bigg)\cdot w\,dx \nonumber\\
&= \int_{S^1} \bigg(\frac{1}{2} \lambda \frac{\partial}{\partial x} \frac{l_{txx} l_x-l_{xx}l_{tx}}{l_x^2}-\pi_t l_x-\pi l_{tx} \bigg)\cdot w\,dx.
\end{align}

\noindent
We now use \eqref{eq: E-L equations for the Clebsch action functional 5} and \eqref{eq: E-L equations for the Clebsch action functional 6} to eliminate the time derivatives in the integrand, which yields

\begin{align}
\bigg \langle\frac{d}{dt} \frac{\delta \ell}{\delta (u,a)}, (w,c) \bigg \rangle 
&= \int_{S^1} \underbrace{\bigg[\frac{1}{2} \lambda \frac{\partial}{\partial x} \frac{-l_x \frac{\partial^2}{\partial x^2}(u l_x)+l_{xx}\frac{\partial}{\partial x}(u l_x)}{l_x^2}+ l_x\frac{\partial}{\partial x} \bigg(\pi u -\frac{1}{2} \lambda \frac{u_{xx}}{l_x}  \bigg)+\pi \frac{\partial}{\partial x}(u l_x) \bigg]}_{A}\cdot w\,dx.
\end{align}

\noindent
Note that the expression $A$ contains the functions $u$, $l$, $\pi$, $\lambda$, and their spatial derivatives. On the other hand we have

\begin{align}
\bigg \langle \text{ad}^*_{(u,a)} \frac{\delta \ell}{\delta (u,a)}, (w,c) \bigg \rangle 
&= \int_{S^1} \bigg[2 \frac{\delta \ell}{\delta u} u_x + u \frac{\partial}{\partial x} \frac{\delta \ell}{\delta u} + \frac{\partial \ell}{\partial a} u_{xxx}  \bigg]\cdot w\,dx \nonumber \\
&= \int_{S^1} \underbrace{\bigg[\bigg(\lambda \frac{\partial}{\partial x} \frac{l_{xx}}{l_x} - 2 \pi l_x\bigg) u_x + u \frac{\partial}{\partial x} \bigg( \frac{1}{2} \lambda \frac{\partial}{\partial x} \frac{l_{xx}}{l_x} -\pi l_x\bigg) + \lambda u_{xxx}  \bigg]}_{B}\cdot \,w\,dx,
\end{align}

\noindent
where in the first equality we used \eqref{eq: Coadjoint action}, and in the second equality we used \eqref{eq: E-L equations for the Clebsch action functional 2} and \eqref{eq: E-L equations for the Clebsch action functional 4}. Note that the expression $B$ contains the functions $u$, $l$, $\pi$, $\lambda$, and their spatial derivatives. After rather tedious, albeit straightforward algebraic manipulations we find that $A+B=0$. Therefore, we have that for all $(w,c)\in \mathfrak{vir}$

\begin{equation}
\bigg \langle \frac{d}{dt} \frac{\delta \ell}{\delta (u,a)}+\text{ad}^*_{(u,a)} \frac{\delta \ell}{\delta (u,a)}, (w,c) \bigg \rangle =0,
\end{equation}

\noindent
which completes the proof, since the inner product is nondegenerate.\\
\end{proof}

\subsubsection{Separable reduced Lagrangian}
\label{sec: Separable reduced Lagrangian}

The variational principle \eqref{eq: Clebsch variational principle} simplifies significantly when one considers separable Lagrangians of the form

\begin{equation}
\label{eq: Separable Lagrangian}
\ell(u,a) = \frac{1}{2}a^2 + \bar \ell(u).
\end{equation}

\noindent
In that case Equations \eqref{eq: E-L equations for the Clebsch action functional 1} and \eqref{eq: E-L equations for the Clebsch action functional 2} imply $\lambda = a = \text{const}$. Treating $a$ as a constant, we can eliminate the variables $\theta$ and $\lambda$ from the action functional \eqref{eq: Clebsch action functional}. Consider the action functional

\begin{equation}
\label{eq: Simplified Clebsch action functional}
S[u,l,\pi] = \int_{t_1}^{t_2}\bigg(\bar \ell(u) + \frac{a}{2} \int_{S^1} \frac{u_x l_{xx}}{l_x}\,dx \bigg)\,dt + \int_{t_1}^{t_2} \int_{S^1} \pi(l_t+ul_x)\,dxdt.
\end{equation}

\noindent
The stationarity condition $\delta S = 0$ with respect to arbitrary variations $\delta u$, $\delta \pi$, and vanishing endpoint variations $\delta l$, yields the variational equations

\begin{subequations}
\label{eq: E-L equations for the simplified Clebsch action functional}
\begin{align}
\label{eq: E-L equations for the simplified Clebsch action functional 1}
\delta u&:      \quad \frac{\delta \bar \ell(u)}{\delta u}  = -\, \pi l_x +  \frac{a}{2} \frac{\partial}{\partial x} \frac{l_{xx}}{l_x},\\
\label{eq: E-L equations for the simplified Clebsch action functional 2}
\delta \pi&:    \quad l_t + u l_x =0, \\
\label{eq: E-L equations for the simplified Clebsch action functional 3}
\delta l&:      \quad \pi_t + \frac{\partial}{\partial x} \bigg(\pi u -\frac{a}{2} \frac{u_{xx}}{l_x}  \bigg)=0.
\end{align}
\end{subequations}

\noindent
It is straightforward to see that the system \eqref{eq: E-L equations for the Clebsch action functional} reduces to \eqref{eq: E-L equations for the simplified Clebsch action functional} for Lagrangians of the form \eqref{eq: Separable Lagrangian}. 

\paragraph{Remark.} The action functional \eqref{eq: Simplified Clebsch action functional} provides a new variational formulation for Equation~\eqref{eq: General family of equations} when the Lagrangian \eqref{eq: Reduced Lagrangian for the family of equations} is considered. For $a=0$ this action functional reduces to the action functional for the dispersionless CH equation ($\alpha=\beta=1$) introduced in \cite{CotterHolmMultisymplectic} and one of the action functionals for the HS equation ($\alpha=0$ and $\beta=1$) described in \cite{HunterZheng1994}. For $\alpha=1$ and $\beta=a=0$ we also obtain a variational principle for the inviscid Burgers' equation.

%%%%%%%%%%%%%%%%%%%%%%%%%%%%%%%%%%%%%%%%%%%%%%%%%%%%%%%%%%%%%%%%%%%%%%%%%%%%%%%%%%%%
%  Inverse map multisymplectic formulation
%%%%%%%%%%%%%%%%%%%%%%%%%%%%%%%%%%%%%%%%%%%%%%%%%%%%%%%%%%%%%%%%%%%%%%%%%%%%%%%%%%%%
\section{Inverse map multisymplectic formulation}
\label{sec:Inverse map multisymplectic formulation}

The action functional and variational principle introduced in Section~\ref{sec: Separable reduced Lagrangian} allow the identification and analysis of a new multisymplectic formulation of the family of equations \eqref{eq: General family of equations}. Multisymplectic geometry provides a covariant formalism for the study of field theories in which time and space are treated on equal footing. Multisymplectic formalism is useful for, e.g., the stability analysis of water waves (see \cite{BridgesMultisymplectic}, \cite{BridgesDerks1999}) or construction of structure-preserving numerical algorithms (see \cite{BridgesReichMultisymplectic}, \cite{MarsdenPatrickShkoller}). The multisymplectic form formula was first proved by Marsden \& Patrick \& Shkoller \cite{MarsdenPatrickShkoller} and provides an intrinsic and covariant description of the conservation of symplecticity law, first introduced by Bridges \cite{BridgesMultisymplectic} in the context of multisymplectic Hamiltonian PDEs. In Section~\ref{eq: Multisymplectic form formula and conservation of symplecticity} we review the multisymplectic geometry formalism and derive the multisymplectic form formula associated with \eqref{eq: Simplified Clebsch action functional}. We further make a connection with Bridges' approach to multisymplecticity in Section~\ref{eq: Multisymplectic Hamiltonian PDE formulation} and determine a multisymplectic Hamiltonian form of the Euler-Lagrange equations \eqref{eq: E-L equations for the simplified Clebsch action functional}.

\subsection{Multisymplectic form formula and conservation of symplecticity}
\label{eq: Multisymplectic form formula and conservation of symplecticity}

The multisymplectic form formula is the multisymplectic counterpart of the fact that in finite-dimensional mechanics, the flow of a mechanical system consists of symplectic maps. It was first proved for first-order field theories in \cite{MarsdenPatrickShkoller}, and later generalized to second-order field theories in \cite{KouranbaevaShkoller}. Since the field theory described by the action functional \eqref{eq: Simplified Clebsch action functional} with the Lagrangian \eqref{eq: Reduced Lagrangian for the family of equations} is second-order, we follow the theory developed in \cite{KouranbaevaShkoller}. For the convenience of the reader, below we briefly review multisymplectic geometry and jet bundle formalism necessary for our discussion. 

Let $X = S^1 \times \mathbb{R}$ represent spacetime and denote the local coordinates by $(x^\mu)=(x^1,x^0)$, where $x^1 \equiv x$ is the spatial coordinate and $x^0 \equiv t$ is time. Define the configuration fiber bundle $\tau_{XY}:Y\longrightarrow X$ as $Y = X \times S^1 \times \mathbb{R} \times \mathbb{R}$. Denote the fiber coordinates by $(y^A) = (y^1, y^2, y^3)$ with $y^1 \equiv l$, $y^2 \equiv u$, and $y^3 \equiv \pi$. Physical fields are sections of the configuration bundle, that is, continuous maps $\phi: X \longrightarrow Y$ such that $\tau_{XY} \circ \phi = \text{id}_X$. In the coordinates $(x^\mu, y^A)$ a field $\phi$ is represented as $\phi(x,t) =(x^\mu, \phi^A(x^\mu))= (x,t,l(x,t),u(x,t),\pi(x,t))$.

For a $k$-th order field theory, the evolution of the field takes place on the $k$-th jet bundle $J^k Y$. The first jet bundle $J^1 Y$ is the affine bundle over $Y$ with the fibers $J^1_y Y$ defined as

\begin{equation}
\label{eq: Definition of the first jet bundle}
J^1_y Y = \Big\{\vartheta: T_{(x,t)} X\rightarrow T_y Y \,\,\big|\,\, T\tau_{XY}\circ \vartheta = \mathrm{id}_{T_{(x,t)}X} \Big\}
\end{equation}

\noindent
for $y\in Y_{(x,t)}$, where the linear maps $\vartheta$ represent the tangent mappings $T_{(x,t)}\phi$ for local sections $\phi$ such that $\phi(x,t)=y$. The local coordinates $(x^\mu,y^A)$ on $Y$ induce the coordinates $(x^\mu,y^A, v^A_{\phantom{A}\mu})$ on $J^1 Y$. Intuitively, the first jet bundle consists of the configuration bundle $Y$, and of the first partial derivatives of the field variables with respect to the independent variables. We can think of $J^1 Y$ as a fiber bundle over $X$. Given a section $\phi:X\longrightarrow Y$, we can define its first jet prolongation 

\begin{equation}
\label{eq: First jet prolongation intrinsic}
j^1 \phi: X \ni (x,t) \longrightarrow T_{(x,t)}\phi \in J^1 Y,
\end{equation} 

\noindent
in coordinates given by 

\begin{equation}
\label{eq: First jet prolongation in coordinates}
j^1\phi(x^\mu)=\bigg(x^\mu,\phi^A(x^\nu),\frac{\partial \phi^A(x^\nu)}{\partial x^\mu}\bigg),
\end{equation}

\noindent
which is a section of the fiber bundle $J^1 Y$ over $X$. For higher-order field theories we consider higher-order jet bundles, defined iteratively by $J^{k+1} Y = J^1(J^k Y)$. We denote the local coordinates on $J^2 Y$ by $(x^\mu,y^A, v^A_{\phantom{a}\mu},w^A_{\phantom{a}\mu\nu})$. The second jet prolongation $j^2\phi:X\longrightarrow J^2 Y$ is given in coordinates by $j^2 \phi(x^\mu)=(x^\mu,\phi^A,\partial \phi^A/\partial x^\mu,\partial^2 \phi^A/\partial x^\mu \partial x^\nu)$. Let $(x^\mu,y^A, v^A_{\phantom{a}\mu},w^A_{\phantom{a}\mu\nu},s^A_{\phantom{a}\mu\nu\sigma})$ denote the coordinates on $J^3 Y$. The third jet prolongation $j^3 \phi$ is defined similar to $j^1 \phi$ and $j^2 \phi$. For more information about the geometry of jet bundles see \cite{Saunders} and \cite{GIMMSY}.

In the jet bundle formalism introduced above, the action functional \eqref{eq: Simplified Clebsch action functional} with the reduced Lagrangian \eqref{eq: Reduced Lagrangian for the family of equations} can be written as 

\begin{equation}
\label{eq: Action functional for field theory}
S[\phi]=\int_\mathcal{U} \mathcal{L}(j^2 \phi)\,d^2x,
\end{equation}

\noindent 
where $\mathcal{U}=S^1\times [t_1,t_2]$, $d^2 x = dx \wedge dt$, and the Lagrangian density $\mathcal{L}:J^2 Y \longrightarrow \mathbb{R}$ is

\begin{equation}
\label{eq: Lagrangian density}
\mathcal{L}(x^\mu,y^A, v^A_{\phantom{a}\mu},w^A_{\phantom{a}\mu\nu}) = \frac{\alpha}{2} (y^2)^2+\frac{\beta}{2} (v^2_{\phantom{a}1})^2+ \frac{a}{2} \frac{v^2_{\phantom{a}1} w^1_{\phantom{a}11}}{v^1_{\phantom{a}1}} + y^3 (v^1_{\phantom{a}0} + y^2 v^1_{\phantom{a}1}).  
\end{equation}

\noindent
Hamilton's variational principle seeks fields $\phi(x,t)$ that extremize $S$, that is,

\begin{equation}
\label{eq:HamiltonPrincipleForFields}
\frac{d}{d\lambda} \bigg|_{\lambda=0} S[\eta_Y^\lambda \circ \phi]=0,
\end{equation}

\noindent
for all $\eta_Y^\lambda$ that keep the boundary conditions on $\partial \mathcal{U}$ fixed, where $\eta_Y^\lambda:Y\longrightarrow Y$ is the flow of a vertical vector field $V$ on $Y$. This leads to the Euler-Lagrange equations

\begin{equation}
\label{eq:Euler-Lagrange Equations for Fields}
\frac{\partial\mathcal{L}}{\partial y^A}(j^2 \phi) - \frac{\partial}{\partial x^\mu} \bigg(\frac{\partial \mathcal{L}}{\partial v^A_{\phantom{a}\mu}}(j^2 \phi)\bigg) + \frac{\partial^2}{\partial x^\mu \partial x^\nu} \bigg(\frac{\partial \mathcal{L}}{\partial w^A_{\phantom{a}\mu\nu}}(j^2 \phi)\bigg)=0,
\end{equation}

\noindent
where Einstein's summation convention is used. With the Lagrangian density \eqref{eq: Lagrangian density}, these Euler-Lagrange equations take the form \eqref{eq: E-L equations for the simplified Clebsch action functional}. For more information on multisymplectic geometry and jet bundle setting of field theories see \cite{GotayHighOrderJet}, \cite{GIMMSY}, \cite{KouranbaevaShkoller}, and \cite{MarsdenPatrickShkoller}.

For a second-order field theory, the multisymplectic structure is defined on $J^3 Y$ (see \cite{KouranbaevaShkoller}). Given the Lagrangian density $\mathcal{L}$ one can define the Cartan $2$-form $\Theta_\mathcal{L}$ on $J^3 Y$, in local coordinates given by 

\begin{align}
\label{eq: Cartan n+1 form for field theory}
\Theta_\mathcal{L} &= \bigg( \frac{\partial \mathcal{L}}{\partial v^A_{\phantom{a}\mu}}-D_\nu \bigg( \frac{\partial \mathcal{L}}{\partial w^A_{\phantom{a}\mu \nu}} \bigg) \bigg) dy^A \wedge dx_\mu + \frac{\partial \mathcal{L}}{\partial w^A_{\phantom{a}\nu \mu}} dv^A_{\phantom{a}\nu} \wedge dx_\mu \nonumber \\
&\quad+ \bigg(\mathcal{L}-\frac{\partial \mathcal{L}}{\partial v^A_{\phantom{a}\mu}} v^A_{\phantom{a}\mu} + D_\nu \bigg( \frac{\partial \mathcal{L}}{\partial w^A_{\phantom{a}\mu \nu}} \bigg)v^A_{\phantom{a}\mu} - \frac{\partial \mathcal{L}}{\partial w^A_{\phantom{a}\nu \mu}} w^A_{\phantom{a}\nu \mu}\bigg) d^2x,
\end{align} 

\noindent
where $d x_\mu = \partial_\mu \righthalfcup d^2x$, i.e., $dx_0=-dx$ and $dx_1=dt$, and the \emph{formal} partial derivative in the direction $x^\nu$ of a function $f: J^2 Y \longrightarrow \mathbb{R}$ is defined in coordinates as

\begin{equation}
\label{eq: Formal partial derivative}
D_\nu f = \frac{\partial f}{\partial x^\nu} + \frac{\partial f}{\partial y^A} v^A_{\phantom{a}\nu} + \frac{\partial f}{\partial v^A_{\phantom{a}\mu}} w^A_{\phantom{a}\mu\nu} + \frac{\partial f}{\partial w^A_{\phantom{a}\sigma \mu}} s^A_{\phantom{a}\sigma\mu\nu}.
\end{equation}

\noindent
For the Lagrangian density \eqref{eq: Lagrangian density}, the Cartan form is

\begin{align}
\label{eq: Cartan form}
\Theta_\mathcal{L} = &-y^3 dy^1 \wedge dx + \bigg(y^3 y^2 - \frac{a}{2}\frac{w^2_{\phantom{a}11}}{v^1_{\phantom{a}1}} \bigg) dy^1 \wedge dt + \bigg(\beta v^2_{\phantom{a}1}  + \frac{a}{2}\frac{w^1_{\phantom{a}11}}{v^1_{\phantom{a}1}} \bigg) dy^2 \wedge dt \nonumber \\
&+ \frac{a}{2}\frac{v^2_{\phantom{a}1}}{v^1_{\phantom{a}1}} dv^1_{\phantom{a}1} \wedge dt + \bigg(\frac{\alpha}{2} (y^2)^2 - \frac{\beta}{2} (v^2_{\phantom{a}1})^2 - \frac{a}{2} \frac{v^2_{\phantom{a}1} w^1_{\phantom{a}11}}{v^1_{\phantom{a}1}} +  \frac{a}{2} w^2_{\phantom{a}11}\bigg) dx \wedge dt.
\end{align}

\noindent
The multisymplectic $3$-form $\Omega_{\mathcal{L}}$ is then defined as the exterior derivative of the Cartan form:  

\begin{align}
\label{eq: Multisymplectic form for field theory}
\Omega_{\mathcal{L}}=d\Theta_\mathcal{L} = &\phantom{+}\,\,dy^1 \wedge dy^3 \wedge dx  \,\,-\,\,  y^3 dy^1 \wedge dy^2 \wedge dt  \,\,-\,\, y^2 dy^1 \wedge dy^3 \wedge dt  \,\,-\,\,  \frac{a}{2}\frac{w^2_{\phantom{a}11}}{(v^1_{\phantom{a}1})^2} dy^1 \wedge dv^1_{\phantom{a}1} \wedge dt \nonumber \\
&+\,\,\frac{a}{2 v^1_{\phantom{a}1}} dy^1 \wedge dw^2_{\phantom{a}11} \wedge dt  \,\,-\,\,  \beta dy^2 \wedge dv^2_{\phantom{a}1} \wedge dt  \,\,+\,\,  \frac{a}{2}\frac{w^1_{\phantom{a}11}}{(v^1_{\phantom{a}1})^2} dy^2 \wedge dv^1_{\phantom{a}1} \wedge dt \nonumber \\
&-\,\,\frac{a}{2 v^1_{\phantom{a}1}} dy^2 \wedge dw^1_{\phantom{a}11} \wedge dt  \,\,-\,\,  \frac{a}{2 v^1_{\phantom{a}1}} dv^1_{\phantom{a}1} \wedge dv^2_{\phantom{a}1} \wedge dt  \,\,+\,\,  \alpha y^2 dy^2 \wedge dx \wedge dt \nonumber \\
&-\,\,\bigg(\beta v^2_{\phantom{a}1} + \frac{a}{2}\frac{w^1_{\phantom{a}11}}{v^1_{\phantom{a}1}} \bigg) dv^2_{\phantom{a}1} \wedge dx \wedge dt  \,\,+\,\, \frac{a}{2}\frac{v^2_{\phantom{a}1} w^1_{\phantom{a}11}}{(v^1_{\phantom{a}1})^2} dv^1_{\phantom{a}1} \wedge dx \wedge dt \nonumber \\
&-\,\, \frac{a}{2}\frac{v^2_{\phantom{a}1}}{v^1_{\phantom{a}1}} dw^1_{\phantom{a}11} \wedge dx \wedge dt  \,\,+\,\, \frac{a}{2} dw^2_{\phantom{a}11} \wedge dx \wedge dt.
\end{align} 

\noindent
Let $\mathcal{P}$ be the set of solutions of the Euler-Lagrange equations, that is, the set of sections $\phi$ satisfying \eqref{eq:HamiltonPrincipleForFields} or \eqref{eq:Euler-Lagrange Equations for Fields}. For a given $\phi \in \mathcal{P}$, let $\mathcal{F}$ be the set of first variations, that is, the set of vector fields $V$ on $Y$ such that $(x,t)\rightarrow \eta^\epsilon_Y\circ \phi(x,t)$ is also a solution, where $\eta^\epsilon_Y$ is the flow of $V$. The multisymplectic form formula for second-order field theories (see \cite{KouranbaevaShkoller}) states that if $\phi \in \mathcal{P}$ then for all $V$ and $W$ in $\mathcal{F}$, 

\begin{equation}
\label{eq:MultisymplecticFormFormula}
\int_{\partial \mathcal{U}} (j^3 \phi)^* \big(j^3 V  \righthalfcup  j^3 W  \righthalfcup  \Omega_{\mathcal{L}}\big) = 0,
\end{equation}

\noindent
where $(j^3 \phi)^*$ denotes the pull-back by the mapping $j^3 \phi$, and $j^3V$ is the third jet prolongation of $V$, that is, the vector field on $J^3 Y$ whose flow is the third jet prolongation of  the flow $\eta^\epsilon_Y$ for $V$, i.e.,

\begin{equation}
\label{eq: First jet prolongation of a vector field intrinsic}
j^3V = \frac{d}{d\epsilon}\bigg|_{\epsilon=0}j^3 \eta^\epsilon_Y.
\end{equation}

\noindent
Consider two arbitrary first variation vector fields $V$, $W$, in the local coordinates $(x^\mu, y^A)$ represented by $(V^\mu(x^\mu, y^A), V^A(x^\mu, y^A))$ and $(W^\mu(x^\mu, y^A), W^A(x^\mu, y^A))$, respectively. Let us work out the form of the formula \eqref{eq: Multisymplectic form for field theory} for $\tau_{XY}$-vertical first variations, i.e., $V^\mu(x^\mu, y^A)=W^\mu(x^\mu, y^A)=0$. Denote the components of $j^3 V$ as $(0,V^A, V^A_{\phantom{a}\mu}, V^A_{\phantom{a}\mu\nu}, V^A_{\phantom{a}\mu\nu\sigma})$, and similarly for $j^3 W$. The multisymplectic form formula then becomes

\begin{equation}
\label{eq: MFF explicit expression}
\int_{\partial \mathcal{U}} -F(x,t)\,dx + G(x,t)\,dt = 0,
\end{equation}

\noindent
with

\begin{align}
\label{eq: F and G in MFF}
F(x,t) = &-W^1 V^3 + W^3 V^1, \nonumber \\
G(x,t) = &-\pi (W^1 V^2 - W^2 V^1) -u (W^1 V^3 - W^3 V^1) -\frac{a}{2}\frac{u_{xx}}{l_x^2} (W^1 V^1_{\phantom{a}1} - W^1_{\phantom{a}1} V^1) \nonumber\\
&+\frac{a}{2 l_x} (W^1 V^2_{\phantom{a}11} - W^2_{\phantom{a}11} V^1) - \beta (W^2 V^2_{\phantom{a}1} - W^2_{\phantom{a}1} V^2)+\frac{a}{2}\frac{l_{xx}}{l_x^2} (W^2 V^1_{\phantom{a}1} - W^1_{\phantom{a}1} V^2) \nonumber \\
&-\frac{a}{2 l_x} (W^2 V^1_{\phantom{a}11} - W^1_{\phantom{a}11} V^2) -\frac{a}{2 l_x} (W^1_{\phantom{a}1} V^2_{\phantom{a}1} - W^2_{\phantom{a}1} V^1_{\phantom{a}1}),
\end{align}

\noindent
where the vector components are evaluated at $j^3 \phi(x,t)$. By applying Stokes' theorem and using the fact that $\mathcal{U}$ is arbitrary, the multisymplectic form formula \eqref{eq: MFF explicit expression} can be rewritten equivalently as the conservation law

\begin{equation}
\label{eq: Conservation of symplecticity}
\frac{\partial}{\partial t} F(x,t) + \frac{\partial}{\partial x} G(x,t) = 0.
\end{equation}

\noindent
This kind of a conservation law was first considered by Bridges \cite{BridgesMultisymplectic}. In Section~\ref{eq: Multisymplectic Hamiltonian PDE formulation} we make a further connection with Bridges' theory and find a multisymplectic PDE form of the Euler-Lagrange equations \eqref{eq: E-L equations for the simplified Clebsch action functional}. 

\subsection{Multisymplectic Hamiltonian PDE formulation}
\label{eq: Multisymplectic Hamiltonian PDE formulation}

Bridges \cite{BridgesMultisymplectic} introduced the notion of multisymplecticity by generalizing the notion of Hamiltonian systems to Partial Differential Equations (PDEs). A multisymplectic structure $(\mathcal{M},\omega,\kappa)$ consists of the phase space $\mathcal{M}=\mathbb{R}^n$, and pre-symplectic 2-forms $\omega$ and $\kappa$, where pre-symplectic means that the 2-forms are closed, but not necessarily nondegenerate. A multisymplectic Hamiltonian system is a PDE of the form

\begin{equation}
\label{eq: Multisymplectic Hamiltonian PDE}
M(z) z_t + K(z) z_x = \nabla H(z),
\end{equation}

\noindent
where $z: X \ni (x,t) \longrightarrow z(x,t) \in \mathcal{M}$ is a function of the spacetime variables $x$ and $t$, $H: \mathcal{M}\longrightarrow \mathbb{R}$ is the Hamiltonian, and $M(z)$, $K(z)$ are $n \times n$ antisymmetric matrices defined by

\begin{equation}
\label{eq: Definition of M and K matrices}
\omega(\overline W, \overline V) \equiv \big\langle M(z) \overline W , \overline V \big\rangle_\mathcal{M}, \qquad\qquad \kappa(\overline W, \overline V) \equiv \big\langle K(z) \overline W , \overline V \big\rangle_\mathcal{M},
\end{equation}

\noindent
where $\overline V$, $\overline W$ are arbitrary vector fields on $\mathcal{M}$, and $\langle \cdot, \cdot \rangle_\mathcal{M}$ is the standard Euclidean inner product on $\mathcal{M}=\mathbb{R}^n$. 

We will use the multisymplectic form formula \eqref{eq: MFF explicit expression} to deduce the multisymplectic Hamiltonian PDE form \eqref{eq: Multisymplectic Hamiltonian PDE} of the Euler-Lagrange equations \eqref{eq: E-L equations for the simplified Clebsch action functional}. We note that for $a>0$ the vector components that appear in \eqref{eq: F and G in MFF} only correspond to the 7 coordinate directions $y^1$, $y^2$, $y^3$, $v^1_{\phantom{a}1}$, $v^2_{\phantom{a}1}$, $w^1_{\phantom{a}11}$, $w^2_{\phantom{a}11}$ on $J^3Y$. We will therefore consider $\mathcal{M}=\mathbb{R}^7$ and denote the coordinates on $\mathcal{M}$ as $(l,u,\pi,\Delta, \Theta, \Xi, \Pi)$. Define the projection map

\begin{equation}
\label{eq: Projection mapping}
\mathbb{F}\mathcal{L}: J^3 Y \ni (x^\mu,y^A, v^A_{\phantom{a}\mu},w^A_{\phantom{a}\mu\nu},s^A_{\phantom{a}\mu\nu\sigma}) \longrightarrow (y^1, y^2, y^3, v^1_{\phantom{a}1}, v^2_{\phantom{a}1}, w^1_{\phantom{a}11}, w^2_{\phantom{a}11}) \in \mathcal{M}.
\end{equation}

\noindent
The suitable entries for the matrices $M(z)$ and $K(z)$ can be read off from \eqref{eq: F and G in MFF} as

\begin{equation}
\label{eq: M and K matrices}
M = 
\begin{pmatrix}
0  & 0 &  1 & 0 & 0 & 0 & 0	\\
0  & 0 &  0 & 0 & 0 & 0 & 0  \\
-1 & 0 &  0 & 0 & 0 & 0 & 0  \\
0  & 0 &  0 & 0 & 0 & 0 & 0  \\
0  & 0 &  0 & 0 & 0 & 0 & 0  \\
0  & 0 &  0 & 0 & 0 & 0 & 0  \\
0  & 0 &  0 & 0 & 0 & 0 & 0
\end{pmatrix},
\qquad\quad
K(z) = 
\begin{pmatrix}
0                                & \pi &  u & \frac{a}{2}\frac{\Pi}{\Delta^2}  & 0     & 0                  & -\frac{a}{2 \Delta}	\\
-\pi                             & 0   &  0 & -\frac{a}{2}\frac{\Xi}{\Delta^2} & \beta & \frac{a}{2 \Delta} & 0  \\
-u                               & 0 &  0 & 0 & 0 & 0 & 0  \\
-\frac{a}{2}\frac{\Pi}{\Delta^2} & \frac{a}{2}\frac{\Xi}{\Delta^2} &  0 & 0 & \frac{a}{2 \Delta} & 0 & 0  \\
0                                & -\beta                          &  0 & -\frac{a}{2 \Delta} & 0 & 0 & 0  \\
0                                & -\frac{a}{2 \Delta}             &  0 & 0 & 0 & 0 & 0  \\
\frac{a}{2 \Delta}               & 0 &  0 & 0 & 0 & 0 & 0
\end{pmatrix}.
\end{equation}

\noindent
With that choice, we have $F(x,t)=\omega(\overline W, \overline V)$ and $G(x,t)=\kappa(\overline W, \overline V)$, where $\overline W = T\mathbb{F}\mathcal{L}\cdot j^3W$ and $\overline V = T\mathbb{F}\mathcal{L}\cdot j^3V$. The Hamiltonian $H$ can be read off from the $dx\wedge dt$ term in \eqref{eq: Cartan form} as

\begin{equation}
\label{eq: Hamiltonian}
H(z) = \frac{\alpha}{2} u^2 - \frac{\beta}{2} \Theta^2 - \frac{a}{2} \frac{\Theta \Xi}{\Delta} +  \frac{a}{2} \Pi.
\end{equation}

\noindent
Below we show that the Euler-Lagrange equations \eqref{eq: E-L equations for the simplified Clebsch action functional} indeed can be given the multisymplectic structure \eqref{eq: Multisymplectic Hamiltonian PDE}.

\begin{thm}
\label{thm: Multisymplectic PDE system}
Suppose $a>0$. Then the Euler-Lagrange equations \eqref{eq: E-L equations for the simplified Clebsch action functional} with the Lagrangian \eqref{eq: Reduced Lagrangian for the family of equations} are equivalent to the multisymplectic Hamiltonian system \eqref{eq: Multisymplectic Hamiltonian PDE} with the matrices \eqref{eq: M and K matrices} and the Hamiltonian \eqref{eq: Hamiltonian}. That is, if $\phi(x,t) = (x,t,l(x,t), u(x,t), \pi(x,t))$ is a solution of \eqref{eq: E-L equations for the simplified Clebsch action functional}, then $z(x,t) = \mathbb{F}\mathcal{L}\circ j^3 \phi(x,t)$ is a solution of \eqref{eq: Multisymplectic Hamiltonian PDE}. Conversely, if $z(x,t)$ is a solution of \eqref{eq: Multisymplectic Hamiltonian PDE}, then $\phi(x,t) =(x,t,z_1(x,t), z_2(x,t), z_3(x,t)) = (x,t,l(x,t), u(x,t), \pi(x,t))$ is a solution of \eqref{eq: E-L equations for the simplified Clebsch action functional}.
\end{thm}

\begin{proof}
Substituting \eqref{eq: M and K matrices} and \eqref{eq: Hamiltonian} in \eqref{eq: Multisymplectic Hamiltonian PDE} yields the system of equations

\begin{subequations}
\label{eq: Explicit form of the multisymplectic Hamiltonian PDE}
\begin{align}
\label{eq: Explicit form of the multisymplectic Hamiltonian PDE 1}
\pi_t + \pi u_x + u \pi_x + \frac{a}{2}\frac{\Pi}{\Delta^2}\Delta_x - \frac{a}{2 \Delta} \Pi_x &= 0, \\
\label{eq: Explicit form of the multisymplectic Hamiltonian PDE 2}
-\pi l_x - \frac{a}{2}\frac{\Xi}{\Delta^2}\Delta_x + \beta \Theta_x + \frac{a}{2 \Delta} \Xi_x &= \alpha u, \\
\label{eq: Explicit form of the multisymplectic Hamiltonian PDE 3}
-l_t - u l_x &=0, \\
\label{eq: Explicit form of the multisymplectic Hamiltonian PDE 4}
-\frac{a}{2}\frac{\Pi}{\Delta^2}l_x + \frac{a}{2}\frac{\Xi}{\Delta^2}u_x + \frac{a}{2 \Delta} \Theta_x &= \frac{a}{2}\frac{\Theta \Xi}{\Delta^2}, \\
\label{eq: Explicit form of the multisymplectic Hamiltonian PDE 5}
-\beta u_x - \frac{a}{2 \Delta} \Delta_x &= -\beta \Theta - \frac{a}{2}\frac{\Xi}{\Delta}, \\
\label{eq: Explicit form of the multisymplectic Hamiltonian PDE 6}
- \frac{a}{2 \Delta} u_x &= - \frac{a}{2}\frac{\Theta}{\Delta}, \\
\label{eq: Explicit form of the multisymplectic Hamiltonian PDE 7}
\frac{a}{2 \Delta} l_x &= \frac{a}{2}.
\end{align}
\end{subequations}

\noindent
Equation~\eqref{eq: Explicit form of the multisymplectic Hamiltonian PDE 7} implies $\Delta = l_x$ and Equation~\eqref{eq: Explicit form of the multisymplectic Hamiltonian PDE 6} implies $\Theta=u_x$. Then, Equations~\eqref{eq: Explicit form of the multisymplectic Hamiltonian PDE 5} and \eqref{eq: Explicit form of the multisymplectic Hamiltonian PDE 4} imply $\Xi=l_{xx}$ and $\Pi=u_{xx}$, respectively. By substituting these identities in the remaining equations \eqref{eq: Explicit form of the multisymplectic Hamiltonian PDE 1}-\eqref{eq: Explicit form of the multisymplectic Hamiltonian PDE 3}, we obtain a system equivalent to \eqref{eq: E-L equations for the simplified Clebsch action functional}, which completes the proof.\\ 
\end{proof}

Bridges \cite{BridgesMultisymplectic} showed that the conservation of symplecticity law

\begin{equation}
\label{eq: Conservation of symplecticity with omega and kappa}
\frac{\partial}{\partial t} \omega(\overline W, \overline V) + \frac{\partial}{\partial x} \kappa(\overline W, \overline V) = 0
\end{equation}

\noindent
is satisfied for solutions $z(x,t)$ of \eqref{eq: Multisymplectic Hamiltonian PDE}, where $\overline W$, $\overline V$ are arbitrary first variations of $z(x,t)$. This is an equivalent statement of \eqref{eq: Conservation of symplecticity}, since if $W$ and $V$ are first variations for \eqref{eq:Euler-Lagrange Equations for Fields}, then $\overline W = T\mathbb{F}\mathcal{L}\cdot j^3W$ and $\overline V = T\mathbb{F}\mathcal{L}\cdot j^3V$ are first variations for \eqref{eq: Multisymplectic Hamiltonian PDE}.

\paragraph{Remark.} Equations \eqref{eq: Multisymplectic Hamiltonian PDE}, \eqref{eq: M and K matrices}, and \eqref{eq: Hamiltonian} provide a new multisymplectic formulation for the family of equations \eqref{eq: General family of equations} with $a>0$. For $a=0$ several special cases can be obtained. If $\beta >0$, then Equations \eqref{eq: Explicit form of the multisymplectic Hamiltonian PDE 4}, \eqref{eq: Explicit form of the multisymplectic Hamiltonian PDE 6}, and \eqref{eq: Explicit form of the multisymplectic Hamiltonian PDE 7} become trivial, and it is enough to consider the variables $z=(l,u,\pi,\Theta)$. The matrices $M$ and $K$ then take the form

\begin{equation}
\label{eq: M and K matrices for a=0}
M = 
\begin{pmatrix}
0  & 0 &  1 & 0 	\\
0  & 0 &  0 & 0   \\
-1 & 0 &  0 & 0   \\
0  & 0 &  0 & 0
\end{pmatrix},
\qquad\quad
K(z) = 
\begin{pmatrix}
0                                & \pi    &  u & 0	   \\
-\pi                             & 0      &  0 & \beta \\
-u                               & 0      &  0 & 0     \\
0                                & -\beta &  0 & 0
\end{pmatrix},
\end{equation}

\noindent
and the Hamiltonian becomes $H(z) = \frac{\alpha}{2} u^2 - \frac{\beta}{2} \Theta^2$. For $\alpha=\beta=1$ this reproduces the multisymplectic structure for the dispersionless CH equation found in \cite{CotterHolmMultisymplectic}, and for $\alpha=0$, $\beta=1$ we obtain a new multisymplectic formulation of the HS equation with $a=0$. If in addition $\beta=0$, then Equation~\eqref{eq: Explicit form of the multisymplectic Hamiltonian PDE 5} also becomes trivial, and a further simplification is possible: we consider the variables $z=(l,u,\pi)$ with the matrices

\begin{equation}
\label{eq: M and K matrices for a=0 and beta=0}
M = 
\begin{pmatrix}
0  & 0 &  1  	\\
0  & 0 &  0    \\
-1 & 0 &  0 
\end{pmatrix},
\qquad\quad
K(z) = 
\begin{pmatrix}
0                                & \pi    &  u	\\
-\pi                             & 0      &  0  \\
-u                               & 0      &  0 
\end{pmatrix},
\end{equation}

\noindent
and the Hamiltonian $H(z) = \frac{\alpha}{2} u^2$. This final simplification provides a multisymplectic formulation for the inviscid Burgers' equation.

%%%%%%%%%%%%%%%%%%%%%%%%%%%%%%%%%%%%%%%%%%%%%%%%%%%%%%%%%%%%%%%%%%%%%%%%%%%%%%%%%%%%
%  SUMMARY
%%%%%%%%%%%%%%%%%%%%%%%%%%%%%%%%%%%%%%%%%%%%%%%%%%%%%%%%%%%%%%%%%%%%%%%%%%%%%%%%%%%%
\section{Summary, open problems and opportunities for future work}
\label{sec:Summary}

In this paper, we have introduced a new type of Clebsch representation that extends the momentum map formulation for fluid dynamics introduced in Holm, Kupershmidt \& Levermore \cite{Holm1983,HoKuLe1983} based on the inverse flow map to the case when the group action governing Lagrangian fluid paths includes the Bott $2$-cocycle in equation \eqref{eq: Bott cocycle}. Physically, this means that linear dispersion with third order spatial derivatives can be included, as required for investigating the multisymplectic structures of the Korteweg-de Vries, Camassa-Holm, and Hunter-Saxton equations. Moreover, the multisymplectic form formula was shown to persist and was derived explicitly for this important class of equations, by using our new type of Clebsch representation, identified in equation \eqref{momap} as the momentum map associated with particle relabeling with group actions which include the Bott $2$-cocycle. In addition, symplecticity was found to be conserved in this new class of flows. Consequently, new types of structure-preserving numerics for soliton equations with linear dispersion can now be developed. 

Multisymplectic integrators are methods that preserve a discrete version of the symplectic conservation law \eqref{eq: Conservation of symplecticity with omega and kappa}. There is numerical evidence that these schemes locally conserve energy and momentum remarkably well (see, e.g., \cite{AscherMcLachlan2004}, \cite{AscherMcLachlan2005}, \cite{BridgesReichMultisymplectic}, \cite{CohenMultisymplecticCH}, \cite{MiyatakeCohen2017}, \cite{WeiPeng2008}, \cite{YuShun2007}, \cite{Zhang2016}, \cite{ZhaoQin2000}), which is a much stronger property than merely global conservation over the whole spatial domain (see \cite{McLachlan1993}). Variational integrators are based on discrete variational principles, which provide a natural framework for the discretization of Lagrangian systems (see, e.g., \cite{LewAVI}, \cite{MarsdenPatrickShkoller}, \cite{MarsdenWestVarInt}, \cite{Pavlov}, \cite{SternDesbrun}, \cite{TyranowskiPHD}, \cite{TyranowskiDesbrunRAMVI}). A discrete action functional can be obtained by discretizing the functional \eqref{eq: Simplified Clebsch action functional} on a spacetime mesh. A variational numerical scheme is then derived by extremizing the discrete action with respect to the discrete set of the values of the fields $l$, $u$, and $\pi$. Variational integrators satisfy a discrete version of the multisymplectic form formula \eqref{eq:MultisymplecticFormFormula}, and are therefore multisymplectic. Moreover, in the presence of a symmetry, they satisfy a discrete version of Noether's theorem, as a consequence of which many of the conservation laws of the continuous system persist. These directions will be explored in future work. They are beyond the scope of the present derivation and formulation. 

Furthermore, the new Clebsch momentum map with the Bott $2$-cocycle in equation \eqref{momap} represents an opportunity to extend the approach in \cite{HolmStochasticFluids2015} of using Clebsch variational principles for introducing noise into continuum mechanics. The new Clebsch momentum map \eqref{momap} will enable us to investigate a new type of interplay among nonlinearity and noise that also includes stochastic linear dispersion. This interplay introduces a class of dynamical problems addressing `wobbling' solitons governed by SPDEs with stochastic mass/label transport. Consider a stochastic deformation of the Euler-Poincar\'e equation \eqref{eq: E-P equation} such that the coadjoint action $\text{ad}^*$ is taken with respect to the perturbed Virasoro algebra element $(u+\xi(x)\circ\dot W(t),a+\eta\circ\dot W(t))$ rather than $(u,a)$, where $\dot W(t)$ denotes the white noise, and the prescribed function $\xi(x)$ and element $\eta\in\mathbb{R}$ represent the spatial correlations of the noise in the advection velocity and the intensity of the noise in $a$, respectively. Because of the definition of the coadjoint action \eqref{eq: Coadjoint action}, the perturbation of $a$ does not have any effect on the equation and can be omitted. The stochastic Euler-Poincar\'e equation will therefore take the form

\begin{equation}
\label{eq: Stochastic E-P equation}
{\sf d}\frac{\delta \ell}{\delta (u,a)} + \text{ad}^*_{(u,a)} \frac{\delta \ell}{\delta (u,a)}\,dt + \text{ad}^*_{(\xi,a)} \frac{\delta \ell}{\delta (u,a)}\circ dW(t) = 0,
\end{equation}

\noindent
where $W(t)$ is the standard Wiener process, ${\sf d}$ denotes stochastic differential with respect to the time variable $t$, and $\circ$ denotes Stratonovich time integration. For the Lagrangian \eqref{eq: Reduced Lagrangian for the family of equations}, we obtain a stochastic deformation of the family of equations \eqref{eq: General family of equations} as

\begin{equation}
\label{eq: Stochastic family of equations}
\begin{split}
{\sf d}(\alpha u - \beta u_{xx}) &+ \big[3 \alpha u u_x - \beta(2 u_x u_{xx} + u u_{xxx})+a u_{xxx} \big] \, dt \\
&+ \big[\alpha (2 \xi_x u + \xi u_x) - \beta(2 \xi_x u_{xx} + \xi u_{xxx})+a \xi_{xxx} \big] \circ dW(t) = 0.
\end{split}
\end{equation}

\noindent
This kind of stochastic deformation has been proposed for the dispersionless Camassa-Holm equation ($\alpha=\beta=1$ and $a=0$ in the equation above), electromagnetic field equations, and various fluid dynamics equations (see \cite{CotterHolm2017}, \cite{CrisanFlandoliHolm2017}, \cite{CrisanHolm2017}, \cite{GayBalmazHolm2017}, \cite{HolmStochasticFluids2015}, \cite{HolmUncertainty2017}, \cite{HolmTyranowskiSolitons}). This approach retains many properties of the unperturbed equations, such as the peaked soliton solutions of the Camassa-Holm equation, and the Kelvin circulation theorem for fluid dynamics. Moreover, for certain functional forms of $\xi(x)$, the introduction of this type of coadjoint transport noise can preserve the deterministic isospectral problem, while introducing stochasticity into the evolution equations for the corresponding eigenfunctions. This stochastic process preserves certain aspects of the inverse scattering methods for determining the soliton solutions of SPDEs, as discussed in \cite{CrisanHolm2017}. The results presented in \cite{CrisanFlandoliHolm2017}, \cite{CrisanHolm2017}, \cite{HolmTyranowskiSolitons} suggest that for smooth initial conditions and for a proper class of the correlation functions $\xi(x)$, the solutions of \eqref{eq: Stochastic family of equations} are likely to retain their spatial regularity. For instance, in the case of spatially uniform noise, with $\xi(x)=\gamma=\text{const}$, if $u(x,t)$ is a solution of \eqref{eq: General family of equations}, then $u(x-\gamma W(t), t)$ is a solution of \eqref{eq: Stochastic family of equations}, which can be easily verified by a direct substitution. 

Under this regularity hypothesis, solutions of Equation~\eqref{eq: Stochastic family of equations} are seen to be critical points of the action functional

\begin{equation}
\label{eq: Simplified stochastic Clebsch action functional}
\begin{split}
S[u,l,\pi] = \int_{t_1}^{t_2}\bar \ell(u)\,dt &+ \frac{a}{2} \int_{S^1}\int_{t_1}^{t_2}\bigg( \frac{u_x l_{xx}}{l_x}\,dt + \frac{\xi_x l_{xx}}{l_x}\circ dW(t) \bigg)\,dx \\
&+ \int_{S^1} \int_{t_1}^{t_2} \pi\Big(\circ {\sf d}l+ul_x \,dt+\xi(x) l_x \circ dW(t)\Big)\,dx,
\end{split}
\end{equation}

\noindent
which is a stochastic deformation of \eqref{eq: Simplified Clebsch action functional}, in which the velocity field $u$ in the reconstruction equations \eqref{eq: Reconstruction equations in terms of the inverse map} is replaced with $u+\xi(x)\circ\dot W(t)$. By following the reasoning presented in Section~\ref{eq: Multisymplectic form formula and conservation of symplecticity}, and ignoring the analytical difficulties arising from introducing stochastic integrals, direct calculation shows that a stochastic version of the multisymplectic form formula \eqref{eq:MultisymplecticFormFormula} holds, and the corresponding stochastic conservation of symplecticity law can be written heuristically as

\begin{equation}
\label{eq: Stochastic conservation of symplecticity}
{\sf d}F(x,t) + \frac{\partial}{\partial x} G(x,t)\,dt + \frac{\partial}{\partial x} \overline G(x,t)\circ dW(t) = 0,
\end{equation}

\noindent
where $F(x,t)$ and $G(x,t)$ have been defined in \eqref{eq: F and G in MFF}, and the function $\overline G(x,t)$ is given by

\begin{equation}
\label{eq: G bar in the stochastic conservation law}
\overline G(x,t) =  -\xi(x) (W^1 V^3 - W^3 V^1) -\frac{a}{2}\frac{\xi_{xx}(x)}{l_x^2} (W^1 V^1_{\phantom{a}1} - W^1_{\phantom{a}1} V^1).
\end{equation}

\noindent
Rigorous derivations and proofs of all these heuristic formulas for the effects of introducing noise this way will be subjects of future work. Of course, it would also be of interest to construct the corresponding stochastic variational and multisymplectic integrators for these investigations in future work. 

Finally, it is worth pointing out that Equation~\eqref{eq: Stochastic family of equations} can be put into Lie-Poisson Hamiltonian form as coadjoint motion,

\begin{equation}
\label{eq: Hamiltonian stochastic PDE form}
{\sf d}m = -J \left( \frac{\delta h}{\delta m} \,dt + \frac{\delta \overline h}{\delta m} \circ dW(t)\right),
\end{equation}

\noindent
where $m = \alpha u - \beta u_{xx}$ is the momentum map, $J= \partial_x m + m \partial_x + a \partial_{xxx}$ is the Lie-Poisson Hamiltonian operator, and the Hamiltonians are $h(m)=\frac{1}{2}\int_{S^1} ({\alpha}u^2+{\beta}u_x^2)\,dx$ and $\overline h(m)=\int_{S^1} \xi(x) \,m\,dx$, for which ${\delta h}/{\delta m}=u$ and ${\delta \overline h}/{\delta m} = \xi(x)$, respectively. We remark that the stochastic KdV form of Equation~\eqref{eq: Stochastic family of equations} with $\alpha=1$ and $\beta=0$, expressed here as a member of the class of stochastic Hamiltonian PDEs in \eqref{eq: Hamiltonian stochastic PDE form}, has also appeared in \cite{ArnaudonPHD}, as  

\begin{equation}
\label{eq: stochastic KdV}
{\sf d}u = -(\partial_x u + u \partial_x + a \partial_{xxx})\big(u \,dt + \xi(x) \circ dW(t)\big)\,.
\end{equation}

\noindent
This new form of the stochastic KdV equation reveals that the class of stochastic Hamiltonian PDEs in \eqref{eq: Hamiltonian stochastic PDE form} involves the interplay between stochastic nonlinear transport and stochastic linear dispersion. The investigation of the dynamical effects arising from these two quite different stochastic mechanisms in the contexts of the KdV and CH equations will be yet another subject of future work.

%%%%%%%%%%%%%%%%%%%%%%%%%%%%%%%%%%%%%%%%%%%%%%%%%%%%%%%%%%%%%%%%%%%%%%%%%%%%%%%%%%%%
%  ACKNOWLEDGMENTS
%%%%%%%%%%%%%%%%%%%%%%%%%%%%%%%%%%%%%%%%%%%%%%%%%%%%%%%%%%%%%%%%%%%%%%%%%%%%%%%%%%%%
\section*{Acknowledgments}

We are very grateful to Alexis Arnaudon and Nader Ganaba for many useful comments, references and stimulating discussions during the present work. During this work, the authors were partially supported by the European Research Council Advanced Grant 267382 FCCA and the UK EPSRC Grant EP/N023781/1 held by DH.

%%%%%%%%%%%%%%%%%%%%%%%%%%%%%%%%%%%%%%%%%%%%%%%%%%%%%%%%%%%%%%%%%%%%%%%%%%%%%%%%%%%%
%  SECTIONS REQUIRED BY PROCEEDINGS OF THE ROYAL SOCIETY - COMMENT OUT FOR ARXIV
%%%%%%%%%%%%%%%%%%%%%%%%%%%%%%%%%%%%%%%%%%%%%%%%%%%%%%%%%%%%%%%%%%%%%%%%%%%%%%%%%%%%
%\subsection*{Data accessibility}
%This paper has no data.
%\subsection*{Competing interests}
%We have no competing interests.
%\subsection*{Authors' contributions}
%Our contributions were equally balanced in a true collaboration.
%\subsection*{Acknowledgements}
%We are very grateful to Alexis Arnaudon and Nader Ganaba for many useful comments, references and stimulating discussions during the present work. 
%\subsection*{Funding statement}
%During this work, the authors were partially supported by the European Research Council Advanced Grant 267382 FCCA and the UK EPSRC Grant EP/N023781/1 held by DH.
%\subsection*{Ethics statement}
%Does not apply.

%%%%%%%%%%%%%%%%%%%%%%%%%%%%%%%%%%%%%%%%%%%%%%%%%%%%%%%%%%%%%%%%%%%%%%%%%%%%%%%%%%%%
%  BIBLIOGRAPHY
%%%%%%%%%%%%%%%%%%%%%%%%%%%%%%%%%%%%%%%%%%%%%%%%%%%%%%%%%%%%%%%%%%%%%%%%%%%%%%%%%%%%

%\bibliographystyle{abbrv}
%\bibliography{bibliography}

\end{document}